\documentclass{article}    
\usepackage{graphicx,amsfonts,amsmath,amssymb,amsthm,enumerate,amsthm,epsfig,color}
\newtheorem{theorem}{Theorem}

\newtheorem{lemma}[theorem]{Lemma}
\newtheorem{proposition}[theorem]{Proposition}
\newtheorem{remark}[theorem]{Remark}
\newtheorem{definition}{Defintion}
\newcommand{\dd}{\;{\rm d}}
\renewcommand{\d}{{\rm d}}
\newcommand{\EE}{{\mathbb E}}
\newcommand{\NN}{{\mathbb N}}
\newcommand{\RR}{{\mathbb R}}
\renewcommand{\SS}{{\mathbb S}}

\newcommand{\xx}{{\bf x}}
\newcommand{\zz}{{\bf Z}}
\newcommand{\vv}{{\bf v}}
\newcommand{\DD}{{\mathcal D}}
\newcommand{\GG}{{\mathcal G}}
\newcommand{\FF}{{\mathcal F}}
\newcommand{\BB}{{\mathcal B}}
\newcommand{\MM}{{\mathcal M}}

\newcommand{\chg}{\textcolor{red}}
\begin{document}
\title{Kinetic models for topological nearest-neighbor interactions}
\maketitle

\begin{center}
{\large Adrien Blanchet} \\
IAST/TSE, Universit\'e Toulouse Capitole -- 21 All\'ee de Brienne 31000 Toulouse, France\\
email: Adrien.Blanchet@ut-capitole.fr\\
$\mbox{}$ \\
{\large         Pierre Degond }\\
Department of Mathematics, Imperial College London, London SW7 2AZ, United Kingdom\\
email: p.degond@imperial.ac.uk
\end{center}
\begin{abstract}
We consider systems of agents interacting through topological interactions. These have  been shown to play an important part in animal and human behavior. Precisely, the system consists of a finite number of particles characterized by their positions and velocities. At random times a randomly chosen particle, the follower, adopts the velocity of its closest neighbor, the leader. We study the limit of a system size going to infinity and, under the assumption of propagation of chaos, show that the limit kinetic equation is a non-standard spatial diffusion equation for the particle distribution function. We also study the case wherein the particles interact with their $K$ closest neighbors and show that the corresponding kinetic equation is the same. Finally, we prove that these models can be seen as a singular limit of the smooth rank-based model previously studied in~\cite{Blanchet_Degond_JSP16}. The proofs are based on a combinatorial interpretation of the rank as well as some concentration of measure arguments.
\end{abstract}

\section{Introduction}
\label{sec_intro}

In the literature on animal behavior including fish \cite{Aoki_JapanFish82}, birds \cite{Lukeman_etal_PNAS10} and even pedestrians \cite{Helbing_Molnar-PRE95}, interactions between individuals are often assumed to be strongly dependent on their relative distance. However it has recently been demonstrated that individuals in bird flocks interact with their nearest neighbors irrespective of their distance~\cite{Ballerini_etal_PNAS08, Cavagna_etal_M3AS10}. More precisely, the authors in~\cite{Ballerini_etal_PNAS08} claim that each bird interacts with between six to eight of its closest neighbors. The authors coined the term of ``topological interaction'' to  refer to such interaction mechanisms and ''topological distance'' to refer to how many other individuals were closer. Even though the reality of this topological interaction has been debated \cite{Eriksson_etal_BehavEcol10}, it now seems to receive consensus following reports that self-propelled particle models based on topological interactions successfully reproduce the observed experimental features \cite{Bode_etal_Interface10, Camperi_etal_InterfaceFocus12, Ginelli_Chate_PRL10}.  

The understanding that birds interact through topological rather than metric distance has generated an intense literature. Topological interactions have been introduced to the  modeling of many natural phenomena, from birds \cite{Hemelrijk_Hildenbrandt_PlosOne11} to pedestrians \cite{Jian_etal_ChinPhysB10}. Mixed metric-topologic interactions have also been proposed \cite{Niizato_etal_Biosystems14, Shang_Bouffanais_EurPhysJB14}. Proof of flocking under topological interaction has been given in \cite{Haskovec_PhysD13, Martin_SystContLett14, Wang_Chen_Chaos16}, while speed to consensus has been shown to depend on the number of interacting neighbors in \cite{Shang_Bouffanais_SciRep14}. Similarly, interactions depending of the behavior of the closest neighbor are probably at play in human interactions such as portfolio theory \cite{Banner_etal_AnnAppProb05, Fernholtz_book02, Ichiba_etal_AnnAppProb11}, competition between coworkers within a firm, risk-taking among traders or aggressive behavior to reach a sexual partner, see for example \cite{davis1993experimental, hill2010risk}. Rank-based dynamics bears similarities with rearrangement (see e.g. \cite{Brenier_TransSocRoyA13} and references therein). 

One of the striking features of topological compared to metric interactions is their scale-invariance property. Indeed, irrespective of the bird concentration within the flock, \cite{Ballerini_etal_PNAS08} proved that the interaction features remain unchanged. In human cognition models, it is also more relevant to restrict interactions to the closest neighbors, as the attention of a subject is intrinsically directed to only a few people around him/her \cite{norman1976memory}. Interactions with close neighbors do not preclude interactions at a longer range, as interactions spread with the conscious or unconscious signals sent by the subjects in response to these interactions. However, issues such as quantify the propagation speed of information sent via topological interactions are poorly understood so far, in particular in the presence of a large number of subjects. This calls for a large-scale theory of topological interactions, or in other words, for the development of meso or macroscopic models of particle systems connected through topological interactions. The aim of this article is specifically to derive a macroscopic model for a large population of particles interacting through topological interactions, starting from a simple microscopic model. To the best of our knowledge, the present work and its predecessor~\cite{Blanchet_Degond_JSP16} are the first to develop a rigorous coarse-graining of topological interactions, with the exception of \cite{Haskovec_PhysD13} which tackled a similar question but for a different kind of interaction, closer to mean-field type interactions.  

More precisely, we will consider a system with interacting mobile agents. At Poisson random times a given agent selects a partner to interact with according to a probability rule which depends on the proximity rank of the partner. The interaction rule is then very simple: the agent changes its velocity to align with that of its partner. The goal of the present work is, by letting the number of agents tend to infinity, to derive an equation for the probability distribution of the agents in phase space (positions, velocities). Thanks to the choice of this simple interaction rule, inspired from earlier work \cite{Carlen_etal_PhysD13, Carlen_etal_M3AS13}, we can concentrate on the mathematical aspects of this derivation. In previous work by the authors~\cite{Blanchet_Degond_JSP16}, the probability rule solely depended on the proximity rank normalized by the total number of interacting partners (or equivalently, on a proximity rank expressed as a percentage). This rule made the number of potential interacting partners tend to infinity as the number of agents also did so. We will refer to this rule as the ''smooth rank-based dynamics''. By contrast, in the present work, we will concentrate on the case where there is only one interaction partner (the nearest one) or a finite number of them (the $K$ nearest ones), even in the limit of the number of agents tending to infinity. We will refer to these dynamics as the ``nearest-neighbor'' or ``$K$-nearest-neighbor'' dynamics. 

In this paper, we show (under the Propagation of Chaos assumption) that the kinetic equation resulting from the nearest-neighbor or $K$-nearest neighbor dynamics is a nonlinear spatial diffusion equation for the particle distribution function in phase space (position, velocity). This equation has a non-classical feature as it involves a spatial anti-diffusion of the density (which is a velocity integral of the distribution function). We will show that this term results from the constraint that the density must satisfy the continuity equation, a constraint resulting from the preservation of the number of particles in the course of an interaction. By contrast, in the previous work~\cite{Blanchet_Degond_JSP16} relative to the smooth rank-based dynamics, we showed (also under the Propagation of Chaos assumption), that the resulting kinetic equation involved a spatially non-local integral equation and that the continuity equation for the density was also satisfied. In the present paper, we also show that we can pass from the non-local integral equation issued from the smooth rank-based dynamics to the nonlinear spatial diffusion equation for the nearest-neighbor dynamics by a process involving a singular concentration of the kernel of the integral equation. This provides a vision of the nearest-neighbor dynamics as a singular limit of the previously studied smooth rank-based dynamics.  

Rank-based dynamics (either smooth or nearest-neighbor) are natural from a mathematical point of view. The proximity rank includes information about the most immediate interaction partners of a given particle. Although the rank is a highly non-linear function of the particle positions and is subject to jumps when two particles cross, it has robust properties such as invariance by permutation of the particle numbers, and has combinatorial interpretation: the probability for an agent to have a rank $k$ with respect to agent $i$ is equal to the probability of having $k-1$ agents between them. Our results strongly rely on this interpretation of the rank, together with some concentration of measure arguments. Rank-based dynamics also exhibit a certain universality, as their kinetic model does not depend on the (finite) number of interacting agents. The kinetic equation that we derive in this work leads to many questions. This equation is a non-standard diffusion equation with mainly unkonwn properties, from the perspective of well-posedness, large-time behavior, regularity, etc. We believe that these questions open fascinating new directions of research in kinetic theory.

Kinetic models of flocking or swarming behavior have been widely investigated in the context of metric interactions. The literature is vast and it is virtually impossible to be exhaustive. Below is a sample of major publications on this topic. Derivation of kinetic models from underlying particle models have been established in \cite{Bolley_etal_AML12, Degond_Motsch_M3AS08, Ha_Tadmor_KRM08, Parisot_Lachowicz_KRM16}. Flocking behavior and pattern formation has been investigated in \cite{Albi_etal_SIAP14, Bertozzi_etal_CMS15, Carrillo_etal_SIMA10, Ha_Liu_CMS09, Ha_Tadmor_KRM08, Motsch_Tadmor_JCP11}.  Equilibria and phase transitions in kinetic flocking models have been studied in \cite{Barbaro_etal_MMS16, Degond_etal_JNLS13, Degond_etal_ARMA15}. Passage from kinetic to hydrodynamic descriptions of flocking has been investigated in \cite{Bellomo_Soler_M3AS12, Degond_etal_JNLS13, Degond_etal_ARMA15, Degond_Motsch_M3AS08, Eftimie_JMB12, Fornasier_etal_PhysD11}. Numerical simulation methods have been put forward in \cite{Albi_Pareschi_MMS13, Gamba_etal_JCP15}.

The organization of the paper is as follows. Section \ref{sec_models} is dedicated to the presentation of the models and our main results, as well as a detailed discussion of them. Section~\ref{sec:1} covers the nearest-neighbor case and the derivation of the kinetic model in this case. Section~\ref{sec:2} extends these results to the case where the particles interact with their $K$ closest neighbors. Section \ref{sec_rbtonearest} develops the proof that the kinetic model of the nearest-neighbor (or $K$-nearest-neighbor) interaction is the limit of the kinetic smooth rank-based interaction model of ~\cite{Blanchet_Degond_JSP16}. A conclusion to this article is given in Section~\ref{conclusion}.

\section{Models and main results}
\label{sec_models}

\subsection{General framework}
\label{subsec_general_framework}

Consider a set of $N$ particles. The particle $i$ is characterized by its position $x_i \in {\mathbb R}^d$ and its velocity $v_i \in {\mathbb R}^d$ where $d\geq 1$ is both the spatial and velocity dimension. The particles $\{ (x_1(t), v_1(t)), \ldots, (x_N(t), v_N(t)) \}$ are subject to the following dynamics 
\begin{itemize}
\item[-] The dynamics is a succession of free-flights and collisions. 
\item[-] During free-flight, particles follow straight trajectories 
\begin{equation*}
\left\{
  \begin{array}{l}
    \dot x_i = v_i, \vspace{.1cm}\\
 \dot v_i = 0.
  \end{array}
\right.
\end{equation*}
\item[-] Let $(\pi_{ij}^N)_{(i,j)\in \{1,\ldots,N\}^2}$ be a stochastic matrix, {\it i.e.} for all $(i,j) \in \{1,\ldots,N\}^2$, $\pi_{ij}^N\in [0,1]$ and for all $i \in \{1,\ldots,N\}$, $\sum_{j=1}^N \pi_{ij}^N=1$. At Poisson random times with a rate equals to $N \lambda(N)$, particles undergo the following collisions process: pick a particle $i$ in $\{1, \ldots, N\}$ with uniform probability ${1}/{N}$ ; then pick a collision partner $j$ with probability $\pi_{ij}^N$ and perform the collision: 
\begin{equation*}
\left\{
  \begin{array}{l}
	(x_i,x_j) \mbox{ remains unchanged, } \vspace{.1cm}\\
 (v_i,v_j) \mbox{ is changed into } (v_j,v_j). 
  \end{array}
\right.
\end{equation*}
\end{itemize}
We will assume that $\pi_{ij}^N$ is a function of the particle positions $(x_1, \ldots, x_N)$, i.e. $\pi_{ij}^N = \pi_{ij}^N (x_1, \ldots, x_N)$ and is permutation invariant, {\it i.e.} for any permutation $\sigma \in \mathfrak{S}_N$ where $\mathfrak{S}_N$ denotes the set of permutations of $\{1,\ldots,N\}$, we have
\begin{equation}\label{eq:permpi}
 \pi_{\sigma(i) \sigma(j)}^N(x_{\sigma(1)}, \ldots, x_{\sigma(N)}) = \pi_{ij}^N(x_1, \ldots, x_N). 
\end{equation}
The function $\lambda(N)$ is an appropriate scaling factor which will be defined and discussed below. The choice to define the Poisson random times as $N \lambda(N)$ allows to avoid heavy notations during the computations. Specifically, $\lambda(N)$ represents the rate of jump per individual.

To simplify the notation, when no confusion is possible, we will denote $\xx:=(x_1,\ldots,x_N)$,  $\vv:=(v_1,\ldots,v_N)$, $Z_i:=(x_i,v_i)$, $\zz:=(Z_1,\ldots,Z_N)$ and $\dd \zz:=\dd x_1 \dd v_1\ldots\dd x_N \dd v_N$. We will also use $f(\!\dd \zz)$ instead of $f(\zz)\dd x$.

The system will be described through the master equation which provides the dynamics of the $N$-particles distribution $f^{(N)}(Z_1,\ldots,Z_N)$, i.e. the joint probability of particles $1$ to $N$ to be at location (in phase space) $Z_1$ to $Z_N$. As shown in~\cite{Blanchet_Degond_JSP16}, thanks to~\eqref{eq:permpi} if $f^{(N)}(t)$ is permutation invariant at time $t=0$ then it is permutation invariant for all times. We study the limit of this dynamics when the number of particles goes to $\infty$. First define the $k$-th marginal ($k \in \{1, \ldots , N \}$) as
\begin{equation*}
  f^{(k)}_N(Z_1,\cdots,Z_N, t):=\int f^{(N)}(\zz)\dd Z_{k+1}\ldots\dd Z_N.
\end{equation*}
We assume that ``Propagation of Chaos'' holds true {\it i.e.} $\forall \zz \in {\mathbb R}^{2nN}$, $\forall t \in [0,\infty)$: 
\begin{equation}\label{eq:assumptionf}
  f^{(N)}(Z_1,\cdots,Z_N, t)=\prod_{\ell=1}^N f^{(1)}_N(Z_\ell,t) + \mbox{ negligible terms as } N \to \infty. 
\end{equation}
This definition is formal as long as we do not specify in which topology the remainder becomes negligible. Making this definition rigorous would require a topology on spaces of functions of an arbitrary number of variables, which is a highly technical endeavor. Usual definitions of propagation of chaos assume that the $k$-th marginal factorizes in the limit $N \to \infty$ for any arbitrary $k$ and so, avoids this problem by considering functions of a fixed number $k$ of variables. Unfortunately, this more tractable definition does not suffice here as the number of particles interacting with a given particle is not bounded. So, we will use the formal definition (\ref{eq:assumptionf}) without further discussion and leave the formulation of a rigorous theory for future work. Assuming that $f^{(1)}_N \to f$ in the limit $N \to \infty$, we aim to derive the equation satisfied by $f$.
\medskip

In this article we are interested in situations in which a population of agents interact not based on their metric distance but on their topological distance. To do so we introduce the following definition:
\begin{definition}[Rank]
Consider $N$ particles located at $x_1, \ldots, x_N$. Consider the $i$-th particle and order the list $\big(|x_j-x_i|\big)_{j=1, \ldots, N, \, j \not = i}$ by increasing order and denote by $r^N(i,j)\in \{1, \ldots,  N-1\}$ the position of the $j$-th item in this list. If two indices $j$ and $j'$ are such that  $|x_j-x_i| = |x_{j'}-x_i|$, then we choose arbitrarily an ordering between these two numbers. We define $r^N(i,i)=0$. Now, we define \emph{the rank of $j$ with respect to $i$} as: 
$$ R^N(i,j) = \frac{r^N(i,j)}{N-1} \quad \in \quad \bigcup_{k=1}^{N-1} \Big\{ \frac{k}{N-1} \Big\}. $$
\end{definition} 
\begin{remark}
The notation $r$ and $R$ have been interchanged compared with \cite{Blanchet_Degond_JSP16}.
\end{remark}
We assume that the interaction probabilities $\pi_{ij}^N$ depend on the particle positions $(x_1, \ldots, x_N)$ only through the rank of $j$ with respect to $i$, i.e.  $\pi_{ij}^N = \Pi (r^N(i,j))$ where the function $\Pi$ will take different forms according to the chosen model. Since the rank of $j$ with respect to $i$ is an intrinsic property of the positions of the particles and does not depend on how they are numbered, we have the following properties of the rank:
\begin{remark} Let $(x_1, \ldots, x_N)$ be a set of $N$ particles.
  \begin{itemize}
  \item[(i)] The rank $r^N(i,j)$, and hence $\pi_{ij}^N$, is a function of $(x_1, \ldots, x_N)$, \emph{i.e.}  $$r^N(i,j) = r^N(i,j)(x_1, \ldots, x_N)\;.$$ More precisely, we consider the rank $r^N(i,j)$ as a function of $L^\infty({\mathbb R^{nN}})$.
  \item[(ii)] The rank is permutation invariant, {\it i.e.} for any permutation $\sigma \in \mathfrak{S}_N$ where $\mathfrak{S}_N$ denotes the set of permutations of $\{1,\ldots,N\}$, we have
\begin{equation*}
 r^N(\sigma(i), \sigma(j))(x_{\sigma(1)}, \ldots, x_{\sigma(N)}) = r^N(i,j)(x_1, \ldots, x_N). 
\end{equation*}
  \end{itemize}
\label{lem:rank} 
\end{remark}
\medskip

We will be considering three different type of collision dynamics which we will describe in the following three sections.


\subsection{ Nearest neighbor dynamics}
\label{subsec_nearest_neighbour}

In this article we first consider the case where the collision takes place with the closest neighbor {\it i.e.}
$$\pi_{ij}^N= \delta_1(r^N(i,j))$$
where $\delta_x$ stands for the Dirac delta centered at $x$.
\medskip

We will prove the following: Assume Propagation of Chaos~\eqref{eq:assumptionf} and assume a specific form for $\lambda(N)$ to be made precise \chg{in~\eqref{eq:lambda}} below. Then, if $f^{(1)}_N \to f$ and $\rho^{(1)}_N:=\int f^{(1)}_N\dd v\to \rho = \int  f \dd v$ as $N \to \infty$, $f$ and $\rho$ satisfy  
\begin{equation}
  \partial_t f(x,v) + v \cdot \nabla_x f(x,v) = Q[f](x,v), 
	\label{eq:kinetic_nn}
\end{equation}
where the collision operator $Q[f]$ is defined by
\begin{equation}
  Q[f]:=  \frac{1}{\rho^{\frac{2}{d}}}    \left[\Delta_x f - \frac{f}{\rho }\Delta_x \rho\right]\;.
	\label{eq:collision_nn}
\end{equation}

Note that we do not intend to give a fully rigorous derivation of this model but of course some assumptions have to be made to ensure that $\rho$ for instance remains positive in the considered domain. In the same spirit the topology in which the convergence takes place will be omitted here and discussed in future works.

\subsection{ $K$-nearest neighbor dynamics}
\label{subsec_k_near_neighbour}

Let $K$ be a fixed integer in $\{1,\ldots,N\}$ and let $(\alpha_k)_{k \in \{1,\ldots,K \}} $ be a given sequence such that $\sum_{k=1}^K\alpha_k=1$. We extend this sequence for $k$ in $\{1,\ldots,N\}$ by taking $\alpha_k=0$ for all $k>K$. We consider a collision rule by which particle  $i$ adopts the $j$th-particle velocity with probability $\alpha_k$ if particle $j$ is particle $i$'s $k$-th nearest neighbor, i.e. the probability $\pi_{ij}$ is given by $\pi_{ij}=\sum_{k=1}^N\alpha_k \,\delta_k \left( r^N(i,j)\right)$. Section~\ref{subsec_nearest_neighbour} on nearest neighbor interaction was dedicated to the case $\alpha_1=1$ and $\alpha_k \equiv 0$ for all $k>1$, i.e. $K=1$. The rule considered here is a generalization to the $K$ nearest neighbors i.e. : 
\begin{equation}o
\pi_{ij}^N (\xx) =  \sum_{k=1}^N\alpha_k\,\delta_{k}(r^N(i,j)(\xx))\;.
\label{eq:piKnn}
\end{equation}
Under the propagation of chaos~\eqref{eq:assumptionf} and upon choosing a convenient scaling function $\lambda(N)$ given in~\eqref{eq:lambdaKn} and discussed in Remark \ref{rem:orderofmag}, we will prove that $f$ satisfies the same equation \eqref{eq:kinetic_nn} as before, with the same collision operator \eqref{eq:collision_nn}. 

\subsection{Smooth rank-based dynamics}
\label{subsec_smooth_rank}

This case has been studied in~\cite{Blanchet_Degond_JSP16}. We recall it here for reference and further comparisons with the other models introduced below. We introduce a function $K$: $R \in [0,1] \mapsto K(R) \in [0,\infty)$ such that 
$$ \int_0^1 K(R) \dd R = 1.$$
We define
$$ K^N (R) = \dfrac{K(R)}{\sum_{k=1}^{N-1} K \Big( \frac{k}{N-1} \Big)}\,,$$
in order to have for any $i \in \{1, \ldots, N\}$: 
$$ \sum_{\substack{j=1\\j\not=i}}^N K^N\big(R^N(i,j) \big) = \sum_{k=1}^{N-1} K^N \left( \frac{k}{N-1} \right) = 1\;. $$
In this way, for any $i \in \{1, \ldots, N\}$, the collection $(\pi_{ij})_{j=1, j \not = i}^N$, where 
$$ \pi_{ij}^N = K^N \big( R^N(i,j) \big), $$
defines a discrete probability measure on the set $\{j\in \{1, \ldots, N\}, \, \, j \not = i\}$.
The collision rule is then as follows
\begin{definition}[Smooth rank-based dynamics]
 Pick a particle $i$ in $\{1, \ldots, N\}$ with uniform probability ${1}/{N}$ ; then pick a particle $j \not = i$ with probability $\pi_{ij}^N = K^N \big( R^N(i,j) \big)$ and perform the collision:
\begin{equation*}
\left\{
  \begin{array}{l}
	(x_i,x_j) \mbox{ remains unchanged, } \vspace{.1cm}\\
 (v_i,v_j) \mbox{ is changed into } (v_j,v_j). 
  \end{array}
\right.
\end{equation*} 
\end{definition}
\medskip

Assuming that $f^{(1)}_N \to f$ and $\rho^{(1)}_N:=\int f^{(1)}_N\dd v\to \rho = \int  f \dd v$, as $N \to \infty$, and assuming $\lambda(N) = 1$,  we prove in~\cite{Blanchet_Degond_JSP16} that $f$ is a solution of the kinetic equation:
\begin{equation}
\label{eq:kinsrbd}
  \frac{\partial f}{\partial t}(x,v)+v\cdot \nabla_x f(x,v)=\rho(x) \int f(x',v)\,K\left( M_\rho (x,|x'-x|)\right)\dd x' - f(x,v), 
\end{equation}
where $M_\rho$ is the partial mass of $\rho$ and is defined by
\begin{equation}
  M_\rho(x,s)=\int_{\{x' \in {\mathbb R}^d \, \, | \, \, |x'-x| \leq s \}} \rho(x')\dd x'\;.
	\label{eq:Mrho}
\end{equation}

\subsection{Discussion}
\label{subsec_discussion}

As mentioned in Sections \ref{subsec_nearest_neighbour} and \ref{subsec_k_near_neighbour}, the kinetic models derived from the nearest neighbor interaction or the $K$-nearest neighbor interactions are the same. On the other hand, these models differ quite significantly from that obtained from the smooth ranked-based dynamics of \cite{Blanchet_Degond_JSP16} as recalled in Section \ref{subsec_smooth_rank}. However, we show in Section \ref{sec_rbtonearest} that the former are limits of the latter when the interaction kernel $K$ concentrates (with a convenient scaling) near zero. In this limit, since only the closest neighbors interact, and these closest neighbors are likely to be spatially close (especially when the density is large) the spatially non-local integral operator appearing in \eqref{eq:kinsrbd} converges to the diffusion operator \eqref{eq:collision_nn}. Note that this diffusion is multiplied by an inverse power of the density. This is easily understood as, when the density is small, the particles are very far apart, resulting in spatial communications between the particles over larger distances, and eventually, into a larger diffusion coefficient. This interpretation is reinforced by the fact that the inverse power of the density depends on the dimension, in the same way as the scaling between the inter-particle distance and the density depends on the dimension. 

In \cite{Blanchet_Degond_JSP16}, we noted that any solution of the smooth rank-based dynamic kinetic model \eqref{eq:kinsrbd} satisfies the mass conservation equation
$$ \partial_t \rho + \nabla_x \cdot (\rho u) = 0, $$
with $\rho = \int f \, \dd v$ and $\rho u = \int f \, v \dd v$. 
The nearest-neighbor kinetic model \eqref{eq:kinetic_nn} also satisfies the mass conservation equation. To see this, it is enough to show that $\int Q(f) \, \dd v = 0$. But we easily check that it is the case. Indeed: 
$$ 
\int Q(f) \, \dd v = \frac{1}{\rho^{\frac{2}{d}}}   \left[\Delta_x \big( \int f \, \dd v \big) - \frac{1}{\rho } \big( \int f \, \dd v \big)  \Delta_x \rho \right] \\= 
\frac{1}{\rho^{\frac{2}{d}}}   \left[\Delta_x \rho -  \Delta_x \rho \right] = 0. 
$$
The collision operator~\eqref{eq:collision_nn} has the form of a spatial diffusion of $f$ but with an anti-diffusion in $\rho$. In fact, this anti-diffusion is exactly the term that needs to be added to turn a pure spatial diffusion  $ \rho^{-\frac{2}{d}}\Delta_x f$ into an operator that conserves mass i.e. that satisfies $\int Q(f) \, \dd v = 0$. 
At the microscopic level, the collision dynamics describes particles communicating their velocity to spatially distant (although close) neighbors. Therefore, information about the velocity distribution propagates to neighboring particles randomly leading to a spatial diffusion of this distribution. However, this spatial diffusion of the velocity distribution is constrained to obey local mass conservation, which is the reason of the anti-diffusion term acting on the density, as stressed above. If there is no spatial variation of the velocity distribution, particles communicating their velocity to their neighbors will not modify the velocity distribution in this neighborhood, which explains why the collision operator vanishes in this case. The well-posedness theory of \eqref{eq:kinetic_nn} is still open but from this remark we can conjecture that the model is well-posed. Indeed, apart from a mass-carrying component, the equation is a spatial diffusion. And the mass carrying component satisfies a continuity equation. So, it seems that the model couples two components each of which solves a well-posed equation. Of course, the coupling is non-trivial and this may result into a lack of well-posedness. This issue will be dealt with in future work.

\section{Nearest neighbor interaction}
\label{sec:1}

\subsection{Master equation and propagation of chaos}\label{mastereq}
\label{subsec_master}

\subsubsection{Master equation} 
\label{subsubsec:master2}

As recalled in~\cite{Blanchet_Degond_JSP16}, when the collisions occur at Poisson times with rate $N \lambda(N)$, the master equation for the $N$-particle distribution function $f^{(N)}$ in weak form is, for all test function $\phi^N:\zz \mapsto \phi^N(\zz)$, given by:
\begin{multline}\label{eq:me}
\int\partial_t f^{(N)}(\zz) \;\phi^N(\zz)\dd \zz = \sum_{i=1}^N \int f^{(N)}(\zz) \,(v_i \cdot \nabla_{x_i})\phi^N(\zz)\dd \zz  \\
+ \lambda(N)  \sum_{\substack{i,j=1\\j\not=i}}^N \iint \delta_{1}[r^N(i,j)(\xx)]\left[\phi^N(Z_1, \ldots, x_i, v_j , \ldots, x_j,v_j, \ldots Z_N)- \phi^N(\zz)\right] \\f^{(N)}(\zz) \dd \zz \;,
\end{multline}
where $\delta_1$ denotes the Dirac Delta centered at $1$. \

\subsubsection{Propagation of chaos and first marginal}
\label{subsubsec:propa}

The following proposition provides the equation for the first marginal $f^{(1)}_N$ under the Propagation of Chaos Assumption ~\eqref{eq:assumptionf}. This equation is given in weak form by using a test function $\phi(Z_1)$ only depending on the first coordinate $Z_1$ in Eq. \eqref{eq:me} and inserting~\eqref{eq:assumptionf} into \eqref{eq:me}. In all this section, we drop the drift term (the first term at the right-hand side of \eqref{eq:me}) as its treatment is classical (see e.g. \cite{Cercignani_etal_Spinger13}). This leads to the following proposition

\begin{proposition}[First marginal equation with propagation of chaos]
Under the propagation of chaos assumption~\eqref{eq:assumptionf},  for all test function $\phi(Z_1)$ and dropping the drift term as well as the negligible terms when $N \to \infty$ in~\eqref{eq:assumptionf}, we have:
\begin{multline}\label{eq:2et}
 \frac{1}{\lambda(N) \, (N-1)} \int\partial_t f^{(1)}_N(Z_1) \;\phi(Z_1)\dd Z_1\\
=  \int\left[\phi(x_1, v_2)-\phi(Z_1)\right]  \left(1- M_{\rho^{(1)}_N}(x_1,|x_2-x_1|)\right)^{N-2}f^{(1)}_N(\!\dd Z_1)   f^{(1)}_N(\!\dd Z_2) \;,
  \end{multline}
with $M_\rho$ given by \eqref{eq:Mrho}. 
\end{proposition}
\begin{proof} As pointed out above, we drop the drift term for simplicity as this term can be handled with classical methods. Taking $\phi^N(\zz)=\phi(Z_1)$ as a test function in~the master equation \eqref{eq:me} and using the permutation invariance as well as a straightforward combinatorial argument shows that 
  \begin{align*}
\frac{1}{\lambda(N) \, (N-1)} & \int\partial_t f^{(N)}(\zz) \;\phi^N(Z_1)\dd \zz\\&\quad=  \int \delta_{1}[r^N(1,2)(\xx)]\,\left[\phi(x_1, v_2)-\phi(Z_1)\right] f^{(N)}(\zz) \dd \zz \\
&\quad + \int\delta_{1}[r^N(2,1)(\xx)]\left[\phi(Z_1)-\phi(Z_1)\right] f^{(N)}(\d \zz) \\
&\quad +(N-2)\int\delta_{1}[r^N(2,3)(\xx)]\left[\phi(Z_1)-\phi(Z_1)\right]  f^{(N)}(\d \zz)\;.
  \end{align*}
The last two terms obviously vanish.

Using the propagation of chaos assumption~\eqref{eq:assumptionf} we obtain
 \begin{multline}\label{eq:1}
  \frac{1}{\lambda(N) \, (N-1)} \int\partial_t f^{(1)}_N(\zz) \;\phi(Z_1)\dd \zz\notag\\
=  \int\left[\phi(x_1, v_2)-\phi(Z_1)\right]   \left\{\int\delta_{1}[r^N(1,2)(\xx)]\prod_{\ell=3}^N \dd \rho^{(1)}_N(x_\ell)  \right\}\notag\\
f^{(1)}_N(\d Z_1)  f^{(1)}_N(\d Z_2) \;.
  \end{multline}
As shown in~\cite{Blanchet_Degond_JSP16}, the integral
\begin{equation*}
  \int\delta_{1}(r^N(1,2)(\xx))\,\prod_{\ell=3}^N \dd \rho^{(1)}_N(x_\ell) 
\end{equation*}
can be interpreted as the expectation of $\delta_{1}(r^N(1,2)(\xx))$ for fixed $(x_1,x_2)$ when $x_3$, \ldots, $x_N$ are drawn randomly and independently with probability $\rho^{(1)}_N$. Using the combinatorial approach of~\cite{Blanchet_Degond_JSP16} to evaluate this probability, we obtain
\begin{align*}
  \EE_{\rho^{(1)}_N} \left[\delta_{1}(r^N(1,2)(\xx))\right] &= \sum_{R=1}^{N-1}\delta_{1}\left({R}\right) {{N-2} \choose {R-1}}\, p^{R-1}\,(1-p)^{N-2-(R-1)}\\
&= (1-p)^{N-2}, 
\end{align*}
with
\begin{equation*}
  p:=M_{\rho^{(1)}_N}(x_1,|x_2-x_1|),
\end{equation*}
which gives the stated result.
\end{proof}
For a given smooth function $\rho$ and $x \in \RR^d$, let $m\in[0,1) \mapsto R_{\rho}(x,m) \in [0,\infty)$ be the inverse function of $r \mapsto M_{\rho}(x,r)$. Note that $R_{\rho}(x,0)=0$. 
\begin{proposition} \label{prop:two}
For any test function $\phi(Z_1)$, we have
\begin{equation}
\label{eq:two}
\hspace{-0.15cm}
\int\partial_t f^{(1)}_N(Z_1) \;\phi(Z_1)\dd Z_1=  \int\left[\phi(x_1, v_2)-\phi(Z_1)\right] \FF_{f^{(1)}_N}(x_1,v_2)f^{(1)}_N(\d Z_1) \dd v_2,
  \end{equation}
where for given smooth functions $(x,v)\mapsto \eta(x,v)$ and $\rho_\eta(x)=\int \eta(x,v)\dd v$, we define
\begin{equation}\label{eq:mug}
  \FF^N_{\eta}(x,v):=\lambda(N) \, (N-1)\int_0^1 \GG_\eta(x,v,m)\,(1-m)^{N-2} \dd m\;,
\end{equation}
and
\begin{equation*}
  \GG_\eta(x,v,m) = \frac{\int_{\SS^{d-1}}\eta(x+R_{\rho_\eta}(x,m)\,\omega,v)\dd \omega}{\int_{\SS^{d-1}}\rho_\eta(x+R_{\rho_\eta}(x,m)\,\omega)\dd \omega}\;.
\end{equation*}
\end{proposition}
\begin{proof}
  Using polar coordinates $x_2=x_1+r\omega$, $r \in [0,\infty)$, $\omega \in \SS^{d-1}$ with $\int_{\SS^{d-1}} \dd \omega=1$, we can rewrite Equation~\eqref{eq:2et} as 
\begin{equation*}
\int\partial_t f^{(1)}_N(Z_1) \,\phi(Z_1)\dd Z_1=  \int\left[\phi(x_1, v_2)-\phi(Z_1)\right] \FF^N_{f^{(1)}_N}(x_1,v_2)f^{(1)}_N(\d Z_1) \dd v_2
  \end{equation*}
where, we temporarily define $\FF^N_\eta$, for a given function $(x,v)\mapsto \eta(x,v)$ by
\begin{equation*}
  \FF^N_{\eta}(x,v)= \lambda(N) \, (N-1)\int \eta(x+r\omega,v) \left(1- M_{\rho_\eta}(x,r)\right)^{N-2}r^{d-1}\dd r\dd \omega\;.
\end{equation*}
We perform the change of variable $m:=M_{\rho_\eta}(x,r)$, so that $r:=R_{\rho_\eta}(x,m)$.  It is then straightforward to see that $\FF^N_{\eta}$ is indeed defined by~\eqref{eq:mug} as the jacobian of the diffeomorphism is
\begin{equation*}
  r^{d-1}\dd r = \frac{\dd m}{\int_{\SS^{d-1}}\rho_\eta(x+R_{\rho_\eta}(x,m)\omega)\dd \omega}\;.
\end{equation*}
\end{proof}

\subsection{Limit equation}
\label{sec:limit}

\subsubsection{Preliminaries}
\label{subsubsec_prelim}

The passage to the limit with different scaling assumptions on $\lambda$ in the various considered models will use a fundamental lemma: Let $a$ and $b$ be two positive parameters. Define the Beta-distribution, $\beta_{a,b}$, the probability density function given for all $s \in [0,1]$ by 
\begin{equation*}
  \beta_{a,b}(s):=\frac{s^{a-1}\,(1-s)^{b-1}}{\BB(a,b)}\;,
\end{equation*}
where the Beta-function $\BB(a,b)$ is defined by
\begin{equation*}
\BB(a,b):=  \int_0^1 u^{a-1}\,(1-u)^{b-1}\dd u\;.
\end{equation*}
\begin{lemma}
\label{lem:2}
  Let $(h^N)_{N \in \NN}$ be a sequence of uniformly bounded smooth functions $[0,1] \to \RR$ and $(b_N)_{N \in \NN}$ be a sequence going to $\infty$. If $h^N(u)$ converges to 0 as $u$ goes to 0 uniformly in $N$ then the expected value of $h^N$ under the Beta-distribution of parameters $(a,b_N)$:
  \begin{equation*}
    \EE_{\beta_{a,b_N}}[h^N]=\frac{\int_0^1 h^N(u)\,u^{a-1}\,(1-u)^{b_N-1}\dd u}{\BB(a,b_N)}
  \end{equation*}
converges to 0 when $N$ goes to $\infty$.
\end{lemma}
\begin{proof}
As $h^N(u)$ converges to 0 as $u$ goes to 0 uniformly in $N$, we have
\begin{equation*}
  \forall \varepsilon >0\,, \exists \delta>0 \mbox{ such that } \forall u \in (0,\delta)\,,\forall N \in \NN\,, |h^N(u)|<\varepsilon\;.
\end{equation*}
Moreover, by assumption, there exists $C$ such that for all $u \in (0,1)$ and $N \in \NN$, $|h^N(u)| <C$. Therefore
\begin{equation}\label{eq:into}
   \EE_{\beta_{a,b_N}}[h^N] \le \varepsilon + C \frac{\int_\delta^1 u^{a-1}\,(1-u)^{b_N-1}\dd u}{\BB(a,b_N)}\;.
\end{equation}
Note that the fraction is closely related to the incomplete regularized Beta-distribution.
On one hand, the function $u \mapsto  u^{a-1}\,(1-u)^{b_N-1}$ is increasing on $[0,u_{a,b_N}]$ and decreasing on $[u_{a,b_N},1]$ where
\begin{equation*}
 u_{a,b_N}:= \frac{a-1}{a+b_N-2}\;,
\end{equation*}
is the maximum on $[0,1]$. As $(b_N)_{N \in \NN}$ goes to $\infty$, for $N$ large enough, $u_{a,b_N}$ is less than $\delta$. As a consequence, for $N$ large enough
\begin{equation}\label{eq:h1}
\int_\delta^1 u^{a-1}\,(1-u)^{b_N-1}\dd u \le \delta^{a-1} (1-\delta)^{b_N-1}\;.
\end{equation}
On the other hand, by standard properties of the Beta-distribution, see~\cite{Askey_Roy_Handbook10}, as $N$ goes to $\infty$ and $a$ is fixed 
\begin{equation}\label{eq:h2}
  \BB(a,b_N)=\int_0^1 u^{a-1}\,(1-u)^{b_N-1}\dd u \sim \Gamma \left(a \right)b_N^{-a}\;,
\end{equation}
where $\Gamma$ is the Gamma-function. 
Collecting~\eqref{eq:h1} and~\eqref{eq:h2} we obtain that there exists $C$ such that
\begin{equation*}
 \frac{\int_\delta^1 u^{a-1}\,(1-u)^{b_N-1}\dd u}{\BB(a,b_N)} \le  C (1-\delta)^{b_N-1}b_N^{a}\;.
\end{equation*}
This term converges to 0 for $N$ large. Coming back to~\eqref{eq:into} gives the stated result.
  \end{proof}

\subsubsection{Case $\lambda(N) = 1$}
\label{subsubsec_lambda=1}

In this section, we show that the kinetic model obtained with $\lambda(N) = 1$ is trivial, {\it i.e.} it involves no contribution of the particle interactions to the final dynamics.
\begin{proposition}[Case $\lambda(N) = 1$]
\label{prop:ff} 
Assume that $(f^{(1)}_N)_{N \in \NN}$ and $(\rho^{(1)}_N)_{N \in \NN}$ converge toward smooth functions $f$ and $\rho$ respectively. If the convergence of $(f^{(1)}_N)_{N \in\NN}$ to $f$ is such that the sequence $h^N=\GG_{f^{(1)}_N}(x_1,v_2,\cdot)-\GG_{f^{(1)}_N}(x_1,v_2,0)$ satisfies the assumptions of Lemma~\ref{lem:2} then we have
$$   
\int\partial_t f(Z_1) \,\phi(Z_1)\dd Z_1 =0\;.
$$
In strong form and after restoring the drift term, the equation for $f$ is the free transport equation:
\begin{equation}
\partial_t f + v \cdot \nabla_x f = 0. 
\label{eq:freetrasnsp}
\end{equation}
\end{proposition}

\begin{proof}
By Lemma~\ref{lem:2} applied to $a=1$, $b_N = N-1$, we have
\begin{equation*}
  \lim_{N \to \infty}\FF^N_{f^{(1)}_N}(x_1,v_2)=\GG_f(x_1,v_2,0)=\frac{f(x_1,v_2)}{\rho(x_1)}\;.
\end{equation*}
Therefore we have 
\begin{equation*}
   \lim_{N \to \infty}\int\partial_t f^{(1)}_N(Z_1) \;\phi(Z_1)\dd Z_1 = \int\left[\phi(x_1, v_2)-\phi(Z_1)\right] \frac{f(x_1,v_2)}{\rho(x_1)}f(\d Z_1) \dd v_2\;.
\end{equation*}
If we exchange $v_1$ and $v_2$, the term $\phi(x_1, v_2)-\phi(x_1, v_1)$ changes of sign while
\begin{equation*}
  f(x_1, v_1) \,\frac{f(x_1,v_2)}{\rho(x_1)}
\end{equation*}
remains unchanged, so that the integral vanishes. 
\end{proof}

\subsubsection{Non-trivial limit}
\label{subsubsec_nontrivial}

Here, we determine what must be the expression of the scaling factor $\lambda(N)$ such that the contribution of the particle interactions in the limit $N \to \infty$ is non-trivial and we determine the corresponding limit model. Before stating the main theorem, we need the following technical lemma.  

\begin{lemma}[Higher order expansion]
\label{lem:exp}
  Let $(x,v)\mapsto \eta(x,v)$ be a smooth function and $\rho_\eta(x):=\int \eta(x,v)\dd v$. Define 
\begin{equation}\label{eq:d}
  \DD[\rho_\eta,\eta](x,v)= \Delta_x \eta(x,v) - \frac{\eta(x,v)\,}{\rho_\eta(x) }\Delta_x \rho_\eta(x)\;.
\end{equation}
	For $m$ small enough, we have for all $(x,v)\in\RR^{2d}$:
\begin{equation}\label{eq:g}
  \GG_\eta(x,v,m)-\GG_\eta(x,v,0)= \frac{m^{\frac{2}{d}} d^{\frac{2}{d}-1}}{2}\frac{1}{\rho_\eta(x)^{1+\frac{2}{d}}} \DD[\rho_\eta,\eta](x,v) + o(m^{\frac{2}{d}})\;,
\end{equation}
where recall that
\begin{equation*}
  \GG_\eta(x,v,0)=\frac{\eta(x,v)}{\rho_\eta(x)}\;.
\end{equation*}
\end{lemma}
\begin{proof}
We need to develop $m \mapsto \GG_\eta(x,v,m)$ for $m$ small to higher order terms. To do so we first compute the next two orders of $\GG_\eta$ expanded in powers of $R_{\rho_\eta}=R_{\rho_\eta}(x,m)$:
\begin{align*}
  \GG_\eta(x,v,m) &- \GG_\eta(x,v,0)=\\
&R_{\rho_\eta} \Bigg[\frac{\int (\omega \cdot \nabla_x\eta )(x+R_{\rho_\eta}\omega,v)\dd \omega}{\int \rho_\eta (x+R_{\rho_\eta}\omega)\dd \omega} \\
&\qquad -\frac{\int \eta (x+R_{\rho_\eta}\omega,v)\dd \omega}{\left( \int \rho_\eta (x+R_{\rho_\eta}\omega)\dd \omega\right)^2} \int \omega \cdot \nabla_x\rho_\eta  (x+R_{\rho_\eta}\omega)\dd \omega \Bigg]\\
&\qquad+ \frac{R_{\rho_\eta}^2}{2}  \Bigg[\frac{\int (\omega \cdot \nabla_x )^2\eta(x+R_{\rho_\eta}\omega,v)\dd \omega}{\int \rho_\eta (x+R_{\rho_\eta}\omega)\dd \omega}  \\
&\qquad -2\frac{\int (\omega \cdot \nabla_x \eta) (x+R_{\rho_\eta}\omega,v)\dd \omega}{\left( \int \rho_\eta (x+R_{\rho_\eta}\omega)\dd \omega\right)^2} \int \omega \cdot \nabla_x \rho_\eta  (x+R_{\rho_\eta}\omega)\dd \omega \\
&\qquad +2\frac{\int \eta (x+R_{\rho_\eta}\omega,v)\dd \omega}{\left( \int \rho_\eta (x+R_{\rho_\eta}\omega)\dd \omega\right)^3} \left(\int \omega \cdot \nabla_x \rho_\eta  (x+R_{\rho_\eta}\omega)\dd \omega\right)^2 \\
&\qquad -\frac{\int \eta (x+R_{\rho_\eta}\omega,v)\dd \omega}{\left( \int \rho_\eta (x+R_{\rho_\eta}\omega)\dd \omega\right)^2} \left(\int (\omega \cdot \nabla_x)^2\rho_\eta  (x+R_{\rho_\eta}\omega)\dd \omega\right)^2 \Bigg]
\end{align*}
Fortunately all the first order terms and the second and third terms in the expression of the second order terms are zero by anti-symmetry. Now, since $\int_{\SS^{d-1}}\omega \otimes \omega\dd \omega={\rm Id}/d$, we obtain when $R_{\rho_\eta} \to 0$:
\begin{equation}\label{eq:g2}
  \GG_\eta(x,v,m) - \GG_\eta(x,v,0)=\frac{R_{\rho_\eta}^2}{2\,d} \left[\frac{\Delta \eta(x,v)}{\rho_\eta(x)} - \frac{\eta(x,v)\Delta_x \rho_\eta(x)}{\rho_\eta(x)^2 }\right] + o\left(R_{\rho_\eta}^2 \right)\;.
\end{equation}
Let us now develop $m \mapsto R_{\rho_\eta}(x,m)$ in powers of $m$. We have for $r$ small enough
\begin{equation*}
  M(x,r):=\int_{|x-x'|<r}\rho_\eta(x')\dd x' = \rho_\eta(x)\, \frac{r^{d}}{d}+o(r^{d})\;.
\end{equation*}
Therefore for $m$ small enough
\begin{equation*}
  R_{\rho_\eta} = \left(\frac{d\,m}{\rho_\eta(x)} \right)^{\frac{1}{d}} + o(m^{\frac{1}{d}})\;,
\end{equation*}
which concludes the proof by inserting this expression in~\eqref{eq:g2}.
\end{proof}

Before stating the theorem, we introduce the following notations: define 
\begin{equation}\label{eq:hN}
h^N(m)=H^N(m)-H^N(0), 
\end{equation}
 where
\begin{equation}\label{eq:hn}
  H^N(m):=\frac{\GG_{f^{(1)}_N}(x_1,v_2,m) - \GG_{f^{(1)}_N}(x_1,v_2,0)}{m^{\frac{2}{d}}}\;,
\end{equation}
and
\begin{equation}\label{eq:h0}
  H^N(0):= \frac{d^{\frac{2}{d}-1}}{2}\frac{1}{\rho^{(1)}_N(x_1)^{\frac{2}{d}+1}} \DD[\rho^{(1)}_N,f^{(1)}_N](x_1,v_2)\;.
\end{equation}
We note that from Lemma \ref{lem:exp}, $h^N(m) \to 0$ as $m \to 0$ for all integer $N$. We assume in the theorem below that this convergence is uniform with respect to $N$. This assumption prevents the occurrence of particle concentrations in the limit $N \to \infty$ which would lead to a non-smooth behavior of $\GG_f(x_1,v_2,m)$ at $m=0$. Now we state the main theorem of this paper: 

\begin{theorem}
Assume that $(f^{(1)}_N)_{N \in \NN}$ and $(\rho^{(1)}_N)_{N \in \NN}$ converge toward smooth functions $f$ and $\rho$ respectively. Assume that the convergence of $(f^{(1)}_N)_{N \in\NN}$ to $f$ is such that the sequence $h^N$ defined by \eqref{eq:hN} satisfies the assumptions of Lemma~\ref{lem:2}. Set $\lambda(N)$ to the value
\begin{equation}\label{eq:lambda}
\lambda(N) =   \frac{1}{(N-1) \frac{d^{\frac{2}{d}-1}}{2}\,\BB(1+\frac{2}{d},N-1)}.
\end{equation}
Then for all test functions $\phi(Z_1)$ we have
\begin{align}
\nonumber
\int\partial_t f(Z_1) &\;\phi(Z_1)\dd Z_1\\
& \hspace{-1cm}= \int\phi(Z_1)\frac{1}{\rho(x_1)^{\frac{2}{d}}}    \left[\Delta_x f(Z_1) - \frac{f(Z_1)}{\rho(x_1) }\Delta_x \rho(x_1)\right]\dd Z_1\;.
\label{eq:kin_nontriv_weak}
  \end{align}
In strong form, introducing the collision operator
\begin{equation}\label{eq:q}
  Q[f]:= \frac{1}{\rho^{\frac{2}{d}}}    \left[\Delta_x f - \frac{f}{\rho }\Delta_x \rho\right]\;,
\end{equation}
and after restoring the drift term, the equation for $f$ reads: 
\begin{equation}\label{eq:kin_nontriv_strong}
\partial_t f + v \cdot \nabla_x f = Q[f]. 
\end{equation}
\label{th:non-trivial_case}
\end{theorem}

Before proving this theorem, we first prove the following intermediate step: 
\begin{proposition}\label{prop:sec} Under the assumptions of Theorem \ref{th:non-trivial_case}, for all test functions $\phi(Z_1)$, we have
  \begin{multline}
    \lim_{N \to \infty}\int_0^1 \left[\GG_{f^{(1)}_N}(x_1,v_2,m) - \GG_{f^{(1)}_N}(x_1,v_2,0)\right]\frac{(1-m)^{N-2}}{\BB(1+\frac{2}{d},N-1)}\dd m\\
		=\frac{d^{\frac{2}{d}-1}}{2}\frac{1}{\rho(x_1)^{\frac{2}{d}+1}} \DD[\rho,f](x_1,v_2). 
  	\label{eq:limGG-GG}
\end{multline}
\end{proposition}

\begin{proof}[Proof of Proposition~\ref{prop:sec}]
Thanks to the assumptions on $h^N$, we can apply Lemma~\ref{lem:2} with $a = 1 + \frac{2}{d}$ and $b_N = N-1$, and obtain
\begin{equation*}
  \lim_{N \to \infty} \int_0^1h^N(m)\dd \beta_{\frac{2}{d}-1,N-1} =0\;.
\end{equation*}
As $\beta_{\frac{2}{d}-1,N-1}$ is a probability this proves that
\begin{equation}
\label{eq:limHN}
  \lim_{N \to \infty} \int_0^1H^N(m)\dd \beta_{1+\frac{2}{d},N-1}=H(0)\;,
\end{equation}
where $H(0)$ is defined from (\ref{eq:h0}) by replacing $f^{(1)}_N$, $\rho^{(1)}_N$ by $f$, $\rho$ respectively.

Using \eqref{eq:hn}, we have
\begin{multline} \label{eq:limHN2}
\int_0^1H^N(m)\dd \beta_{1+\frac{2}{d},N-1}\\=\int_0^1 \left[\GG_{f^{(1)}_N}(x_1,v_2,m) - \GG_{f^{(1)}_N}(x_1,v_2,0)\right]\frac{(1-m)^{N-2}}{\BB(1+\frac{2}{d},N-1)}\dd m\;.
\end{multline}
Inserting \eqref{eq:limHN2} and the expression of $H(0)$ given by \eqref{eq:h0} into \eqref{eq:limHN} leads to \eqref{eq:limGG-GG}. \end{proof}

\begin{proof}[Proof of Theorem \ref{th:non-trivial_case}]
We start from \eqref{eq:two} where $\FF^N$ is defined by~\eqref{eq:mug}. Using \eqref{eq:limGG-GG} and owing to the fact that the term involving $\GG_{f^{(1)}_N}(x_1,v_2,0)$ vanishes upon integration with respect to $(Z_1,v_2)$ (exactly like in the proof of Proposition \ref{prop:ff}), we deduce that, for $N$ large, up to negligible terms when $N \to \infty$, we have
\begin{align*}
 \int\partial_t f(Z_1)& \;\phi(Z_1)\dd Z_1\\
&= \lambda(N) \, (N-1)\,  \frac{d^{\frac{2}{d}-1}}{2}\,\BB(1+\frac{2}{d},N-1)\\
&\int\left[\phi(x_1, v_2)-\phi(x_1, v_1)\right] \frac{f(Z_1)}{\rho(x_1)^{\frac{2}{d}+1}} \DD[\rho,f](x_1,v_2)\dd Z_1\dd v_2. 
  \end{align*}
Using Definition~\eqref{eq:d} of $\DD$, we see that the second term of $f(Z_1) \, \DD[\rho,f](x_1,v_2)$ is symmetric by exchange of $(v_1,v_2)$. Since $\left[\phi(x_1, v_2)-\phi(x_1, v_1)\right]$ is anti-symmetric under this exchange, this second term vanishes after integration. Therefore, the only non-zero term comes from the first term of $\DD$, which leads to (up to negligible terms when $N \to \infty$) to
\begin{align*}
 \int\partial_t f(Z_1) &\;\phi(Z_1)\dd Z_1\\
&= \lambda(N) \, (N-1)\,  \frac{d^{\frac{2}{d}-1}}{2}\,\BB(1+\frac{2}{d},N-1)\\
&\int\phi(Z_1)\frac{1}{\rho(x_1)^{\frac{2}{d}}}    \left[\Delta_x f(Z_1) - \frac{f(Z_1)}{\rho(x_1) }\Delta_x \rho(x_1)\right]\dd Z_1\;,
  \end{align*}
and with the choice \eqref{eq:lambda}, this leads to \eqref{eq:kin_nontriv_weak}. It readily seen, applying Green's formula, that the strong form of the equation is \eqref{eq:kin_nontriv_strong}, upon inserting back the drift term. \end{proof} 

\begin{remark}
We note that $\lambda(N) \to \infty$ as $N \to \infty$. In particular, in the case $d=2$, $\lambda(N)$ can be easily computed and has value  $\lambda(N)= 2N$. This means that the limit is non-trivial only if the number of collisions per unit time is larger than for the standard kinetic scale (which corresponds to $\lambda(N)=1$). 
\label{rem:lambda}
\end{remark}

\section{$K$-nearest neighbor dynamics}
\label{sec:2}

\subsection{Master equation} 
\label{subsec:master}

Let $K$ be a fixed integer in $\{1,\ldots,N\}$ and a sequence $(\alpha_k)_{k \in \{1,\ldots,K}$ such that $\sum_{k=1}^K\alpha_k=1$ being given. We extend this sequence for $k$ in $\{1,\ldots,N\}$ by taking $\alpha_k=0$ for all $k>K$. The particle $i$ will adopt the velocity of its $j$th-nearest neighbor with a probability $\pi_{ij}^N$ given by \eqref{eq:piKnn}. Section~\ref{sec:1} was devoted to the case $\alpha_1=1$ and $\alpha_j \equiv 0$ for all $j>1$. 

As shown in~\cite{Blanchet_Degond_JSP16}, when the collisions occur at Poisson times with rate $\lambda(N) N$, the master equation of the $N$-particle distribution function $f^{(N)}$ in weak form is, for all test function $\phi^N:\zz \mapsto \phi^N(\zz)$:
\begin{multline}\label{eq:me2}
\int\partial_t f^{(N)}(\zz) \;\phi^N(\zz)\dd \zz = \sum_{i=1}^N \int f^{(N)}(\zz) \,(v_i \cdot \nabla_{x_i})\phi^N(\zz)\dd \zz  \\
+ \lambda(N) \sum_{i=1}^N \sum_{\substack{j=1\\j\not=i}}^N \iint \sum_{k=1}^K \alpha_k \delta_{k}[r^N(i,j)(\xx)]\big[\phi^N(Z_1, \ldots, x_i, v_j , \ldots, x_j,v_j, \ldots Z_N)\\
- \phi^N(\zz)\big] f^{(N)}(\zz) \dd \zz \;.
\end{multline}

\subsection{Propagation of chaos}
\label{subsec_propchaosKn}

\begin{proposition}[First marginal equation with propagation of chaos]
Under the propagation of chaos assumption~\eqref{eq:assumptionf} we have for all test function $\phi$
\begin{multline}\label{eq:2et2}
 \frac{1}{\lambda(N) \, (N-1)} \int\partial_t f^{(1)}_N(Z_1) \;\phi(Z_1)\dd Z_1=  \sum_{k=1}^K \alpha_k{{N-2} \choose {k-1}}\\\int\left[\phi(x_1, v_2)-\phi(Z_1)\right] M_{\rho^{(1)}_N}(x_1,|x_2-x_1|)^{k-1}\left(1- M_{\rho^{(1)}_N}(x_1,|x_2-x_1|)\right)^{N-k-1}\\f^{(1)}_N(\d Z_1)   f^{(1)}_N(\d Z_2)\;.
  \end{multline}
\end{proposition}
\begin{proof} As before we have 
  \begin{align*}
\frac{1}{\lambda(N) \, (N-1)} & \int\partial_t f^{(N)}(\zz) \;\phi^N(Z_1)\dd \zz\\&\quad=  \sum_{k=1}^K\alpha_k\int\sum_{k=1}^K\delta_{k}[r^N(1,2)(\xx)]\,\left[\phi(x_1, v_2)-\phi(Z_1)\right] f^{(N)}(\zz) \dd \zz \;,
  \end{align*}
which leads, by the propagation of chaos assumption~\eqref{eq:assumptionf} to
 \begin{multline}\label{eq:1}
 \frac{1}{\lambda(N) \, (N-1)} \int\partial_t f^{(1)}_N(\zz) \;\phi(Z_1)\dd \zz\notag\\
=  \sum_{k=1}^K\alpha_k\int\left[\phi(x_1, v_2)-\phi(Z_1)\right]   \left\{\int\delta_{k}[r^N(1,2)(\xx)]\prod_{\ell=3}^N \dd \rho^{(1)}_N(x_\ell)  \right\}\notag\\
f^{(1)}_N(\d Z_1)  f^{(1)}_N(\d Z_2) \;,
  \end{multline}
	up to negligible terms as $N \to \infty$. 
The computation of the integral
\begin{equation*}
  \int\delta_{k}(r^N(1,2)(\xx))\,\prod_{\ell=3}^N \dd \rho^{(1)}_N(x_\ell) 
\end{equation*}
is slightly different from before as here (see \cite{Blanchet_Degond_JSP16} for details)
\begin{align*}
  \EE_{\rho^{(1)}_N} \left[\delta_{k}(r^N(1,2)(\xx))\right] &= \sum_{R=1}^{N-1}\delta_{k}\left({R}\right) {{N-2} \choose {R-1}}\, p^{R-1}\,(1-p)^{N-2-(R-1)}\\
&= {{N-2} \choose {k-1}}\, p^{k-1}\,(1-p)^{N-k-1}, 
\end{align*}
with
\begin{equation*}
  p:=M_{\rho^{(1)}_N}(x_1,|x_2-x_1|),
\end{equation*}
which gives the stated result.
\end{proof}
The probability considered in Section~\ref{sec:1} was the Bernoulli distribution. A similar result can be proved with a more general binomial distribution: For a given smooth function $\rho$ and $x \in \RR^d$, let $m\in[0,1) \mapsto R_{\rho}(x,m) \in [0,\infty)$ be the inverse function of $r \mapsto M_{\rho}(x,r)$. We also define the binomial distribution with parameters $n \in \NN$ and $p \in [0,1]$, denoted $\MM(n,p)$, as the probability on the discrete set $\{0, 1, \ldots, N \}$ given by the probability mass function:
\begin{equation*}
 \mu(k;n,p)={\mathbb P}(X=k)= {n \choose k}\,p^k\,(1-p)^{n-k}\;.
\end{equation*}
We note that, by integration by parts, we have 
\begin{equation}
\label{eq:intmu}
\chg{(n+1)} \int_0^1 \mu(k;n,p) \, \dd p  = 1. 
\end{equation}

\begin{proposition} \label{prop:two2}
We have
\begin{multline*}
 \int\partial_t f^{(1)}_N(Z_1) \;\phi(Z_1)\dd Z_1\\
=  \int\left[\phi(x_1, v_2)-\phi(Z_1)\right] \FF_{f^{(1)}_N,K}^N (x_1,v_2)f^{(1)}_N(\d Z_1) \dd v_2
  \end{multline*}
where for given smooth functions $(x,v)\mapsto \eta(x,v)$ and $\rho_\eta(x)=\int \eta(x,v)\dd v$, we define
\begin{equation}
\label{eq:mug2}
  \FF^N_{\eta,K}(x,v):=\lambda(N) \,\int_0^1 \GG_\eta(x,v,m)\,  \nu_K(m)\dd m\;,
\end{equation}
with
\begin{equation*}
  \nu_K(m):= (N-1) \sum_{k=1}^K \alpha_k\mu(k-1;N-2,m)
\end{equation*}
and
\begin{equation*}
  \GG_\eta(x,v,m) = \frac{\int_{\SS^{d-1}}\eta(x+R_{\rho_\eta}(x,m)\,\omega,v)\dd \omega}{\int_{\SS^{d-1}}\rho_\eta(x+R_{\rho_\eta}(x,m)\,\omega)\dd \omega}\;.
\end{equation*}
\end{proposition}
The proof is identical and will not be repeated here. Also note that $\nu_K(m) , \dd m$ is a probability on [0,1]. Indeed, thanks to \eqref{eq:intmu}, we have
  \begin{align*}
  \int_0^1 \nu_K(m) \, \dd m &=(N-1) \sum_{k=1}^K \alpha_k \int \mu(k-1;N-2,m) \dd m = \sum_{k=1}^K \alpha_k =1\;.
\end{align*}

\subsection{Limit equation}
\label{subsec_limiteqKn}

\subsubsection{Case $\lambda(N)=1$}
\label{subsubsec_lambda1Kn}

\begin{proposition}[Case $\lambda(N)=1$]\label{prop:ff2} Assume that $(f^{(1)}_N)_{N \in \NN}$ and $(\rho^{(1)}_N)_{N \in \NN}$ converge toward smooth functions $f$ and $\rho$ respectively. If the convergence of $(f^{(1)}_N)_{n \in\NN}$ to $f$ is such that the sequence $h^N=\GG_{f^{(1)}_N}(x_1,v_2,\cdot)-\GG_{f^{(1)}_N}(x_1,v_2,0)$ satisfies the assumptions of Lemma~\ref{lem:2} then for all test functions $\phi$ we have
  \begin{equation*}
   \int\partial_t f(Z_1) \,\phi(Z_1)\dd Z_1 =0\;.
\end{equation*}
In strong form and after restoring the drift term, the equation for $f$ is the free transport equation \eqref{eq:freetrasnsp}. 
\end{proposition}
\begin{proof}
  As $\nu_K(m)$ is a convex linear combination of binomial distributions, we can apply Lemma~\ref{lem:2} to $\mu(k-1;N-2,m)$ for any $k$ and this leads to the result. 
\end{proof}

\subsubsection{Non-trivial case}
\label{subsubsec_nontrivKn}

Like in section \ref{subsubsec_nontrivial}, we now determine the scaling factor $\lambda(N)$ for the particle interactions to have a non-trivial contribution in the limit $N \to \infty$ and we determine the corresponding limit model.

\begin{theorem}
Assume that $(f^{(1)}_N)_{N \in \NN}$ and $(\rho^{(1)}_N)_{N \in \NN}$ converge toward smooth functions $f$ and $\rho$ respectively. Assume that the convergence of $(f^{(1)}_N)_{N \in\NN}$ to $f$ is such that the sequence $h^N$ defined by \eqref{eq:hN} satisfies the assumptions of Lemma~\ref{lem:2}. Set $\lambda(N)$ to the value
\begin{equation}\label{eq:lambdaKn}
\lambda(N) =   \frac{1}{(N-1) \frac{d^{\frac{2}{d}-1}}{2}\,\sum_{k=1}^K \alpha_k\, {N-2 \choose k-1} \, \BB(k+\frac{2}{d},N-k)}.
\end{equation}
Then for all test functions $\phi(Z_1)$, $f$ satisfies Eq. \eqref{eq:kin_nontriv_weak}.  In strong form and after restoring the drift term, the equation for $f$ is given by \eqref{eq:kin_nontriv_strong}
\label{th:non-trivial_case_Kn}
\end{theorem}

The result still relies on Lemma~\ref{lem:exp} and the following analog of Proposition \ref{prop:sec}: 

\begin{proposition}
\label{prop:sec2} Under the assumptions of Theorem \ref{th:non-trivial_case_Kn}, for all test functions $\phi(Z_1)$, we have
  \begin{multline*}
    \lim_{N \to \infty}\int_0^1 \left[\GG_{f^{(1)}_N}(x_1,v_2,m) - \GG_{f^{(1)}_N}(x_1,v_2,0)\right] \\
	 \times	\frac{\sum_{k=1}^K \alpha_k \, {N-2 \choose k-1} \, m^{k-1} (1-m)^{N-k-1}}{\sum_{k=1}^K \alpha_k \, {N-2 \choose k-1} \, \BB(k+\frac{2}{d},N-k)}\dd m\\=\frac{d^{\frac{2}{d}-1}}{2}\frac{1}{\rho(x_1)^{\frac{2}{d}+1}} \DD[\rho,f](x_1,v_2)\;.
  \end{multline*}
\end{proposition}

The proof Proposition \ref{prop:sec2} is similar to that of Proposition \ref{prop:sec} and is left to the reader. The deduction of Theorem \ref{th:non-trivial_case_Kn} from Proposition \ref{prop:sec2} is exactly the same as that of Theorem \ref{th:non-trivial_case} from Proposition \ref{prop:sec}. 

\begin{remark}
We note the following property which can be proven by integration by parts: for all $k \in {\mathbb Z}$, for all $a \in [-k,\infty)$, for all $b \in [k,\infty)$, we have: 
$$ \BB(a+k, b-k) = \BB(a,b) \, {a+k-1 \choose k} \, {b-1 \choose k}^{-1}. $$
Thanks to this property, we have for $k \in {\mathbb N}$, $k \geq 2$: 
$$ {N-2 \choose k-1} \, \BB(k+\frac{2}{d},N-k)  = {\frac{2}{d}+k-1 \choose k-1} \, \BB(1+\frac{2}{d},N-1). $$
Thus, 
$$ \lambda(N) =  \frac{1}{\sum_{k=1}^K {\frac{2}{d}+k-1 \choose k-1} \, \alpha_k} \, \, \,  \frac{1}{(N-1) \frac{d^{\frac{2}{d}-1}}{2}\, \BB(1+\frac{2}{d},N-1)}.
$$
Therefore, we find the same scaling of $\lambda(N)$ as in the nearest neighbor interaction case, up to the multiplication by the constant factor $(\sum_{k=1}^K {\frac{2}{d}+k-1 \choose k-1} \, \alpha_k)^{-1} $. 
\label{rem:orderofmag}
\end{remark}

\section{From smooth to nearest-neighbor interaction}
\label{sec_rbtonearest}

In this section we investigate the connection between the smooth rank-based interaction developed in~\cite{Blanchet_Degond_JSP16} and the nearest neighbor (or $K$ nearest-neighbor) interaction considered here. We show that when the kernel $K(m)$ concentrates near $m=0$, we pass from \eqref{eq:kinsrbd} to \eqref{eq:kinetic_nn}. As the kernel $K$ concentrates, it needs to be rescaled in the appropriate way. More precisely, we introduce a rescaling parameter $\varepsilon$ which we will tend to zero and the following rescaled kernel
\begin{equation}
\label{eq:rescK}
  K^\varepsilon(m)=\frac{1}{\varepsilon^{1+\frac{2}{d}}} \, K^0\Big(\frac{m}{\varepsilon}\Big)\;,
\end{equation}
and we assume that $K_0$ is normalized such that 
\begin{equation}
\label{eq:normK0}
\frac{d^{\frac{2}{d}-1}}2\int m^{\frac{2}{d}}\,K^0(m)\dd m = 1. 
\end{equation}
\begin{proposition}
\label{prop:rbtonn}
  The kinetic nearest-neighbor interaction model \eqref{eq:kinetic_nn} is the limit when $\varepsilon$ goes to 0 of the kinetic smooth rank-based interaction model \eqref{eq:kinsrbd} with interaction kernel $K=K^\varepsilon$ given by \eqref{eq:rescK} with normalization given by \eqref{eq:normK0}.
\end{proposition}
\begin{proof}
The weak form of the kinetic smooth rank-based interaction model \eqref{eq:kinsrbd} obtained in~\cite{Blanchet_Degond_JSP16} is given, for any test function $\phi(Z_1)$ by
\begin{multline*}
  \int \partial_t f(Z_1,t) \phi(Z_1)\dd Z_1 = \int [\phi(Z_1)-\phi(x_1,v_2)] \\ K^\varepsilon \left(M_\rho(x_1,|x_2-x_1|)\right)f(\dd Z_1)\,f(\dd Z_2)\;. 
\end{multline*}
Passing to polar coordinates in $x_2$, i.e. writing $x_2 = x_1 + r \omega$, $r=|x_2-x_1|$, $\omega = \frac{x_2-x_1}{|x_2-x_1|}$, this equation reads
\begin{multline*}
  \int \partial_t f(Z_1,t) \phi(Z_1)\dd Z_1 = \int [\phi(Z_1)-\phi(x_1,v_2)]  f(x_1+r\,\omega,v_2) \\
	K^\varepsilon \left(M_\rho(x_1,r)\right)f(\dd Z_1)r^{d-1}\dd r\dd \omega\;. 
\end{multline*}
Using the change of coordinates $m=M_\rho(x_1,r)$ the inverse function of which is $r=R_\rho(x_1,m)$ and the jacobian of which is
\begin{equation*}
  \dd m=\rho(x_1+r\omega)r^{d-1}\dd r, 
\end{equation*}
we obtain
\begin{multline*}
  \int \partial_t f(Z_1,t) \phi(Z_1)\dd Z_1 = \int [\phi(Z_1)-\phi(x_1,v_2)]  \GG(x_1,v_2,m)\\K^\varepsilon \left(m\right)f(\dd Z_1)\dd m\dd \omega\dd v_2\;,
\end{multline*}
where
\begin{equation*}
  \GG(x,v,m):=\frac{f(x+R_\rho(x,m)\,\omega,v)}{\rho(x_1+R_\rho(x_1,m)\,\omega)}\;.
\end{equation*}
Let us look at the limit of
$
  \int K^\varepsilon(m)\GG(x,v,m)\dd m
$
when $\varepsilon$ goes to 0. By Lemma~\ref{lem:exp}, we have for small $\varepsilon$
\begin{align*}
  \int K^\varepsilon(m)\GG(x,v,m)\dd m&=\frac{1}{\varepsilon^{1+\frac{2}{d}}} \int K^0\Big(\frac{m}{\varepsilon}\Big) \,\GG(m)\dd m =\frac{1}{\varepsilon^{\frac{2}{d}}} \int K^0(m)\,\GG(\varepsilon m)\dd m\\
&=\frac{1}{\varepsilon^{\frac{2}{d}}} \int K^0(m)\left[\GG(x,v,0)+ \frac{\varepsilon^{\frac{2}{d}}}2\,m^{\frac{2}{d}}\,\GG^{(2)}(x,v,0)+o(\varepsilon^{\frac{2}{d}})\right]\dd m
\end{align*}
where by~\eqref{eq:g}
\begin{equation*}
  \GG^{(2)}(x,v,0):=\frac{d^{\frac{2}{d}-1}}{2}\frac1{\rho(x)^{\frac{2}{d}+1}} \DD[\rho,f](x,v)\;.
\end{equation*}
As the first term is anti-symmetric in the transform $(v_1,v_2) \to (v_2,v_1)$, we obtain at the limit $\varepsilon \to 0$ and thanks again to Lemma~\ref{lem:exp}:
\begin{multline*}
  \int \partial_t f(Z_1,t) \phi(Z_1)\dd Z_1 \\
	= \int m^{\frac{2}{d}}\,K^0(m)\dd m \, \,  \int [\phi(Z_1)-\phi(x_1,v_2)]  \GG^{(2)}(x_1,v_2,0) f(\dd Z_1)\dd v_2 \\
	= \left( \frac{d^{\frac{2}{d}-1}}2\int m^{\frac{2}{d}}\,K^0(m)\dd m \right)   \, \,  \int [\phi(Z_1)-\phi(x_1,v_2)] \frac{1}{\rho(x_1)^{\frac{2}{d}+1}} \hspace{0.9cm} \mbox{}\\ \DD[\rho,f] (x_1,v_2) f(\dd Z_1)\dd v_2\;. 
\end{multline*}
The passage to the strong form is identical to that of Section \ref{subsubsec_nontrivial} and we  obtain
\begin{equation*}
  \partial_t f = \left( \frac{d^{\frac{2}{d}-1}}2\int m^{\frac{2}{d}}\,K^0(m)\dd m \right) Q[f], 
\end{equation*}
where $Q$ is defined by~\eqref{eq:q}. This concludes the proof assuming that $K_0$ is normalized according to  \eqref{eq:normK0}. 
\end{proof}

\section{Conclusion}
\label{conclusion}

In this paper, we have put forward a particle interaction model where particles interact with their nearest neighbor. We have shown that the large particle limit under the Propagation of Chaos assumption is a spatial diffusion for the particle distribution function corrected by an anti-diffusion term acting on the spatial density. We have shown that the appearance of this anti-diffusion term depending on the spatial density results from the fact that the interactions are mass-preserving. We have also considered a model in which particles interact with their $K$ nearest neighbors, for a fixed value of $K$, showing that the corresponding kinetic model is the same as in the nearest neighbor interaction case. Finally, we have linked this work with the previous article \cite{Blanchet_Degond_JSP16} where smooth rank-based dynamics were considered and shown that the kinetic nearest-neighbor model can be recovered from the former through a singular limit involving a scaling of the interaction kernel. 
    
The kinetic models obtained here, as in \cite{Blanchet_Degond_JSP16}, are novel. Their mathematical theory is entirely open: proving existence and uniqueness of solutions, investigating large-time behavior, equilibria and other qualitative properties of the solutions will require the establishment of an appropriate mathematical framework. In parallel, more elaborate physical interaction models (such as the Cucker-Smale \cite{Cucker_Smale_IEEE07} or Motsch-Tadmor \cite{Motsch_Tadmor_JCP11} models) should also be considered. One important question is to investigate how the present results are robust to the introduction of noise in the interaction dynamics. Indeed, noise play an important part of many flocking models in nature. Finally, adequate numerical methods for the kinetic models must be developed and the assessment of the kinetic models against the particle ones in realistic situations should be carefully documented so that these models can be used in practice.

\medskip
\noindent
{\bf Acknowledgements:} PD acknowledges support by the Engineering and Physical Sciences 
Research Council (EPSRC) under grants no. EP/M006883/1, EP/N014529/1 and EP/P013651/1, by the Royal Society and the Wolfson Foundation through a Royal Society Wolfson Research Merit Award no. WM130048 and by the National Science Foundation (NSF) under grant no. RNMS11-07444 (KI-Net). PD is on leave from CNRS, Institut de Math\'ematiques de Toulouse, France. AB acknowledges support from the IDEX-UNITI-ANR-11-IDEX-0002-Transversalit\'e-Projet MUSE "Multi-disciplinary study of emergence phenomena".

\medskip
\noindent
{\bf Data statement:} no new data were collected in the course of this research.

\end{document}